\documentclass[11pt]{iopart}
\expandafter\let\csname equation*\endcsname\relax
\expandafter\let\csname endequation*\endcsname\relax
\usepackage{graphicx}
\usepackage{amssymb}
\usepackage{amsmath}
\usepackage{amsthm}
\usepackage{cite}
\usepackage{hyperref}
\usepackage{epstopdf}
\usepackage{tikz}
\usepackage{pstricks} 
\usepackage{comment}
\usepackage{tikz}
\usetikzlibrary{graphs,patterns,decorations.markings,arrows,matrix}

\newtheorem{definition}{Definition}
\newtheorem{thm}{Theorem}
\newtheorem{prop}{Proposition}
\newtheorem{lemma}{Lemma}
\newtheorem{cor}{Corollary}
\newtheorem{remark}{Remark}

\usepackage[font=small, labelfont=bf, width=0.8\textwidth]{caption}

\DeclareMathOperator{\sign}{sign}
\DeclareMathOperator{\dd}{d}

\def\be{\begin{equation}}
\def\ee{\end{equation}}

\makeatletter
\newcommand*{\textoverline}[1]{$\overline{\hbox{#1}}\m@th$}
\makeatother

\newcommand\bra[1]{{\langle#1|}}
\newcommand\ket[1]{{|#1\rangle}}

\begin{document}


\title[Koornwinder polynomials and the stationary open mASEP]
{Koornwinder polynomials and the stationary multi-species asymmetric exclusion process with open boundaries}

\author{$^1$Luigi~Cantini, $^2$Alexandr~Garbali, $^3$Jan~de~Gier and $^4$Michael~Wheeler}
\address{${}^1$Laboratoire de Physique Th\'eorique et Mod\'elisation (CNRS UMR 8089),
Universit\'e de Cergy-Pontoise, F-95302 Cergy-Pontoise, France, \\ ${}^{2,3,4}$Australian Research Council Centre of Excellence for Mathematical and Statistical Frontiers (ACEMS), School of Mathematics and Statistics, The University of Melbourne, VIC 3010, Australia}
\eads{$^1$luigi.cantini@u-cergy.fr,$^2$alexandr.garbali@unimelb.edu.au,$^3$jdgier@unimelb.edu.au, $^4$wheelerm@unimelb.edu.au}

\begin{abstract}
We prove that the normalisation of the stationary state of the multi-species asymmetric simple exclusion process (mASEP) is a specialisation of a Koornwinder polynomial. As a corollary we obtain that the normalisation of mASEP factorises as a product over multiple copies of the two-species ASEP.
\end{abstract}

\maketitle 



\section{Introduction}
The asymmetric simple exclusion process (ASEP) is a Markov chain of hopping particles. It describes the asymmetric diffusion of hard-core particles along a one-dimensional chain with $n$ sites, and is one of the best studied models in non-equilibrium statistical mechanics  \cite{ASEP1,ASEP2,Derrida98,Schuetz00,GolinelliMallick06,BEreview,Mallicklectures}.   In continuous time its transition rates are given as in Figure~\ref{fig:aseprates}.
\begin{figure}[h]
\begin{center}
\begin{tikzpicture}[scale=1]
 \draw(0.5,0) -- (10.5,0);
 \draw(0.5,-0.5) -- (0.5,0.5);
 \draw(10.5,-0.5) -- (10.5,0.5);
 
  \foreach \i in {0,...,9}
    {
    \pgfmathsetmacro{\xcoord}{1+\i))}
    \pgfmathsetmacro{\ycoord}{0}
    \draw[fill,color=white,draw=black] (\xcoord,\ycoord) circle (0.15);
    }
 \draw[fill,color=black] ({1+1},{0}) circle (0.15);
\draw[fill,color=black] ({1+4},{0}) circle (0.15);
\draw[fill,color=black] ({1+7},{0}) circle (0.15);
\draw[fill,color=black] ({1+8},{0}) circle (0.15);

\draw[thick,->] (5,0.5) arc (160:20:0.5) ;
\draw[thick,->] (5,-0.5) arc (-20:-160:0.5) ;
\node[scale=1] at (5.5,1.1){$t$};
\node[scale=1] at (4.5,-1.2){$1$};

\draw[thick,->] (0,0.5) arc (160:20:0.5) ;
\draw[->,thick] (1,-0.5) arc (-20:-160:0.5);
\node[scale=1] at (0.5,1.1){$\alpha$};
\node[scale=1] at (0.5,-1.2){$\gamma$};

\draw[thick,->] (10,0.5) arc (160:20:0.5) ;
\draw[->,thick] (11,-0.5) arc (-20:-160:0.5);
\node[scale=1] at (10.5,1.1){$\beta$};
\node[scale=1] at (10.5,-1.2){$\delta$};
\end{tikzpicture}
\end{center}
\caption{Rates of the asymmetric simple exclusion process with open boundary conditions.}
\label{fig:aseprates}
\end{figure}
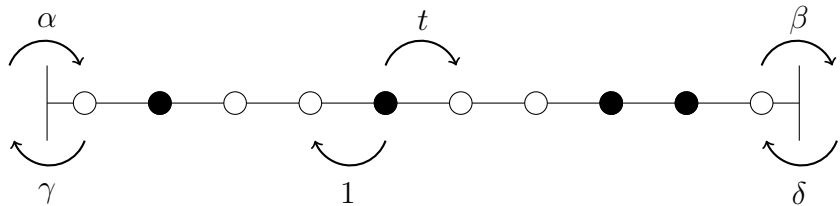

Configurations are labelled by binary strings $\mu=(\mu_1,\ldots,\mu_n)$ where $\mu_i=-1,+1$ for a hole or particle respectively. In terms of this notation the bulk rates are
\be
\begin{split}
(\ldots,-1,+1,\ldots) &\mapsto (\ldots,+1,-1,\ldots)\quad \text{with rate } 1,\\
(\ldots,+1,-1,\ldots) &\mapsto (\ldots,-1,+1,\ldots)\quad \text{with rate } t,
\end{split}
\ee
while at the boundary we have
\be
\begin{split}
(-1,\ldots) &\mapsto (+1,\ldots)\quad \text{with rate } \alpha,\\
(+1,\ldots) &\mapsto (-1,\ldots)\quad \text{with rate } \gamma,
\end{split}
\ee
and
\be
\begin{split}
(\ldots,+1) &\mapsto (\ldots,-1)\quad \text{with rate } \beta,\\
(\ldots,-1) &\mapsto (\ldots,+1)\quad \text{with rate } \delta.
\end{split}
\ee

At late times the ASEP exhibits a relaxation towards a non-equilibrium stationary state. In the presence of two boundaries at which particles are injected and extracted with given rates, the bulk behaviour at stationarity is strongly dependent on the injection and extraction rates. The corresponding phase diagram as well as various physical quantities have been determined by exact methods \cite{Derrida98,Schuetz00,DEHP,gunter,sandow,EsslerR95,PASEPstat1,PASEPstat2,BEreview,GierE1,GierE2}.

It is well known that the stationary state of the ASEP with open boundaries is related to the theory of Askey-Wilson polynomials \cite{UchiSW,CorteelW11}. In this paper we extend this connection to the multi-variable case of Koornwinder polynomials.

\subsection{Multi species}
This process can be generalised to include many species (or colours/gray scales) of particles such as depicted in Figure~\ref{fig:maseprates}.
\begin{figure}[h]
\begin{center}
\begin{tikzpicture}[scale=1]
 \draw(0.5,0) -- (10.5,0);
 \draw(0.5,-0.5) -- (0.5,0.5);
 \draw(10.5,-0.5) -- (10.5,0.5);
 
  \foreach \i in {0,...,9}
    {
    \pgfmathsetmacro{\xcoord}{1+\i))}
    \pgfmathsetmacro{\ycoord}{0}
    \draw[fill,color=white,draw=black] (\xcoord,\ycoord) circle (0.15);
    }
 \draw[fill,color=black!40] ({1},{0}) circle (0.15);
 \draw[fill,color=black] ({1+1},{0}) circle (0.15);
 \draw[fill,color=black!60] ({1+2},{0}) circle (0.15);
 \draw[fill,color=black!40] ({1+3},{0}) circle (0.15);
\draw[fill,color=black!80] ({1+4},{0}) circle (0.15);
\draw[fill,color=black!60] ({1+5},{0}) circle (0.15);
\draw[fill,color=black] ({1+7},{0}) circle (0.15);
\draw[fill,color=black!40] ({1+8},{0}) circle (0.15);
\draw[fill,color=black!60] ({1+9},{0}) circle (0.15);

\draw[thick,<->] (5,0.5) arc (160:20:0.5) ;
\draw[thick,<->] (5,-0.5) arc (-20:-160:0.5) ;
\node[scale=1] at (5.5,1.1){$t$};
\node[scale=1] at (4.5,-1.2){$1$};

\draw[thick,<->] (0,0.5) arc (160:20:0.5) ;
\draw[<->,thick] (1,-0.5) arc (-20:-160:0.5);
\node[scale=1] at (0.5,1.1){$\alpha$};
\node[scale=1] at (0.5,-1.2){$\gamma$};

\draw[thick,<->] (10,0.5) arc (160:20:0.5) ;
\draw[<->,thick] (11,-0.5) arc (-20:-160:0.5);
\node[scale=1] at (10.5,1.1){$\beta$};
\node[scale=1] at (10.5,-1.2){$\delta$};
\end{tikzpicture}
\end{center}
\caption{Rates of the multi-species asymmetric simple exclusion process with open boundary conditions.}
\label{fig:maseprates}
\end{figure}
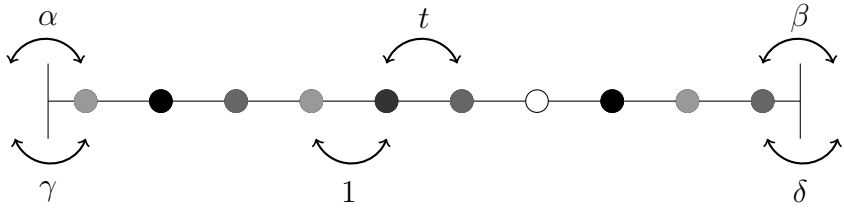
In this case configurations are labelled by strings $\mu=(\mu_1,\ldots,\mu_n)$ where $\mu_i\in \{-r,\ldots,-1,0,+1,\ldots,+r\}$ and each label represents a particular species of particles. In terms of this notation the bulk rates of the model we are interested in are given by
\be
\begin{split}
(\ldots,\mu_i,\mu_{i+1},\ldots) &\mapsto (\ldots,\mu_{i+1},\mu_i,\ldots)\quad \text{with rate }
\left\{\begin{array}{ll}
1 & \text{ if } \mu_i < \mu_{i+1},\\
t & \text{ if } \mu_i > \mu_{i+1}.
\end{array}\right.
\end{split}
\ee
At the boundary there are several possibilities that one could choose, here we take nonzero rates for the following events,
\be
\begin{split}
(-m,\ldots) &\mapsto (+m,\ldots)\quad \text{with rate } \alpha,\\
(+m,\ldots) &\mapsto (-m,\ldots)\quad \text{with rate } \gamma,
\end{split}
\ee
for $m\in\{1,\ldots,r\}$ and likewise 
\be
\begin{split}
(\ldots,+m) &\mapsto (\ldots,-m)\quad \text{with rate } \beta,\\
(\ldots,-m) &\mapsto (\ldots,+m)\quad \text{with rate } \delta.
\end{split}
\ee
Other boundary rates are considered, for example, in \cite{AyyerLS,AyyerLS2,CrampeMRV,CrampeFRV}. See also Section~\ref{se:genbc} of the current paper for generalisations.

While much is known about the stationary state of the multi-species ASEP with periodic boundary conditions \cite{FerrariM1,FerrariM2,EvansFM,ProlhacEM,AritaM,AritaAMP}, not much is known for open boundary conditions beyond rank 1, i.e. beyond the case of two-species \cite{Uchi,AyyerLS,AyyerLS2}. Multi-species totally asymmetric exclusion processes with inhomogeneous hopping rates were considered in \cite{Cantini09,AyyerL12,AritaM,AyyerL14}.

\section{Transition and transfer matrix for rank 1}
\label{se:Tmat1}

The state space $\mathcal{H}_{n,1}:= \left(\mathbb{C}^{3}\right)^{\otimes n}$ of the rank 1 asymmetric exclusion process is spanned by the standard basis
\be
\ket{\mu} = \ket{\mu_1,\ldots,\mu_n},\qquad \mu_i\in\{-1,0,1\}.
\ee
This case is often called the two-species ASEP in the literature \cite{Cantini15,CorteelW,CorteelMW}. In our setup the boundary conditions preserve the number of 0's and the ``single species'' ASEP is simply the sector of the rank 1 ASEP without 0's. 

\subsection{Continuous time transition matrix}
The transition rules can be conveniently encoded in a transition matrix acting on this basis. In this setup the master equation for the time evolution of a state $\ket{\Psi(t)}$ for the ASEP in continuous time is given by
\be
\frac{\dd }{\dd t} \ket{\Psi(t)} = L \ket{\Psi(t)},\qquad \ket{\Psi(t)} = \sum_{\mu} \psi_{\mu} \ket{\mu},
\ee
with a transition matrix or generator $L$ constructed below. Let us firstly describe the standard ASEP with one species of particles, and holes. This case corresponds to rank $r=1$ but we initially leave out colour $0$. For two sites, on the basis $\{\ket{-1,-1},\ket{-1,+1},\ket{+1,-1},\ket{+1,+1}\}$ the transition matrix $L$ is given by
\be
L  =
\begin{pmatrix}
-\alpha & \gamma \\
\alpha & -\gamma
\end{pmatrix}
\otimes 
\mathbb{I}_2
+
\begin{pmatrix}
0 & 0 & 0 & 0 \\
0 & -1 & t & 0 \\
0 & 1 & -t & 0 \\
0 & 0 & 0 & 0
\end{pmatrix}
+
\mathbb{I}_2
\otimes
\begin{pmatrix}
-\delta &  \beta\\
\delta & -\beta
\end{pmatrix},
\ee
where $\mathbb{I}_k$ is the $k\times k$ identity matrix. On $n$ sites we write the transition matrix for the ASEP as 
\be
L= L_0 + \sum_{i=1}^{n-1} L_i + L_{n},
\label{eq:L}
\ee
where
\begin{align}
L_0 &= \begin{pmatrix}
-\alpha & \gamma \\
\alpha & -\gamma
\end{pmatrix}
\otimes 
\mathbb{I}_2^{\otimes n},\nonumber\\
L_i &= \mathbb{I}_2^{\otimes(i-1)}\otimes 
\begin{pmatrix}
0 & 0 & 0 & 0 \\
0 & -1 & t & 0 \\
0 & 1 & -t & 0 \\
0 & 0 & 0 & 0
\end{pmatrix}
 \otimes \mathbb{I}_2^{\otimes(n-i-1)},
 \label{eq:transm}\\
L_n &= \mathbb{I}_2^{\otimes n}
\otimes
\begin{pmatrix}
-\delta &  \beta\\
\delta & -\beta
\end{pmatrix}.\nonumber
\end{align}

\subsection{Discrete time transfer matrix}
In this section we define a Yang--Baxter integrable discrete time transfer matrix \cite{baxterbook} that commutes with the continuous time transition matrix $L$ \eqref{eq:L}. To do so we need to define a R-matrix based on the quantum group $U_{t^{1/2}}(A^{(1)}_{1})$, as well as boundary K-matrices that incorporate the boundary transition rates.

\subsubsection{R-matrix and buk hopping rates.}

We first introduce the functions
\begin{align}
b^+  &= \displaystyle \frac{t(1-x)}{t-x},  &  b^-  &= t^{-1} b^+ , \nonumber\\
c^+ & =1-b^+ , & c^- &=1-b^- .
\end{align}
In terms of these functions the R-matrix associated to the ASEP becomes
\begin{align}
\check{R}(x)=
\begin{pmatrix}
1 & 0 & 0 & 0
\\
0 & c^- &  b^+ & 0
\\
0 & b^- &c^+ & 0
\\
0 & 0 & 0 & 1
\end{pmatrix}.
\end{align}
The bulk hopping transition matrix $L_i$ \eqref{eq:transm} is given by the derivate of the matrix $\check{R}(x)$,
\begin{align}
L_i =(1-t) \check{R}'_i(1).
\end{align}
Reintroducing the ``second class particle'' label 0 the transition rates are encoded in the R-matrix based on $U_{t^{1/2}}(A^{(1)}_{2})$ and given by
\begin{align}
\check{R}(x)=
\begin{pmatrix}
1 & 0 & 0 & 0 & 0 & 0 & 0 & 0 & 0
\\
0 & c^- & 0 & b^+ & 0 & 0 & 0 & 0 & 0
\\
0 & 0 & c^- & 0 & 0 & 0 & b^+ & 0 & 0
\\
0 & b^- & 0 & c^+ & 0 & 0 & 0 & 0 & 0
\\
0 & 0 & 0 & 0 & 1 & 0 & 0 & 0 & 0
\\
0 & 0 & 0 & 0 & 0 & c^- & 0 & b_+ & 0
\\
0 & 0 & b^- & 0 & 0 & 0 & c^+ & 0 & 0
\\
0 & 0 & 0 & 0 & 0 & b^- & 0 & c^+ & 0
\\
0 & 0 & 0 & 0 & 0 & 0 & 0 & 0 & 1
\end{pmatrix}, 
\end{align}
and again 
\be
L_i =(1-t) \check{R}'_i(1),\qquad \check{R}_i(x)=\mathbb{I}_3^{\otimes(i-1)}\otimes \check{R}(x)\otimes \mathbb{I}_3^{\otimes(n-i-1)}.
\ee 

It will also be useful to supply the R-matrix with two indices since it acts on two tensor components, we will write $\check{R}_{i,i+1}(x)=\check{R}_i(x)$, as well as $\check{R}_{i,i+2}(x)=P_{i+1,i+2}\check{R}_i(x) P_{i+1,i+2}$ and so on, where we used the permutation matrix
 \be
 P_{i,i+1}=\mathbb{I}_3^{\otimes(i-1)}\otimes P \otimes \mathbb{I}_3^{\otimes(n-i-1)},
 \ee 
 and
\be
P=\sum_{a,b} E^{(a,b)}\otimes E^{(b,a)},
\ee 
defined in terms of the matrix units $E^{(a,b)}$ which have a single non-zero entry equal to $1$ at position $(a,b)$.
 For the construction of the discrete time transfer matrix we will use the unchecked R-matrix
\be
R(x)=P\check{R}(x).
\ee

\subsubsection{K-matrix and boundary hopping rates.}

Following the standard boundary integrability approach \cite{Skl} (see also \cite{CrampeRV})
we encode the boundary events using K-matrices. Define $h_0(a,c,x)$ and $h_n(b,d,x)$ as
\be
h_0(a,c,x)=(x+a)(x+c),\qquad h_n(b,d,x)=(bx+1)(dx+1),
\ee 
and two additional constants $t_0=-a c$ and $t_n=-b d$ with $a,b,c$ and $d$ being new parameters related to the hopping rates as we will see below.
The boundary matrix $K_0(x)$ is given by
\be\label{K01}
K_0(x) = \mathbb{I}_3 +\frac{1-x^2}{h_0(a,c,x)} 
\begin{pmatrix}
t_0 & 0 & -1  \\
0 & 0 & 0 \\
-t_0 & 0 & 1
\end{pmatrix},
\ee
The matrix $K_0(x)$ is stochastic as its columns add up to $1$ and for suitable values of the parameters the off-diagonal elements are non-negative. The right hand side boundary matrix $K_n(x)$ is defined as
\be\label{Kn1}
K_n(x) = \mathbb{I}_3 - \frac{1-x^2}{h_n(b,d,x)} 
\begin{pmatrix}
1 & 0 & -t_n  \\
0 & 0 & 0 \\
-1 & 0 & t_n
\end{pmatrix}.
\ee
The boundary transitions matrices $L_0$ and $L_n$ in \eqref{eq:transm} are obtained from the two matrices $K_0(x)$ and $K_n(x)$ by taking the derivative
 \begin{align}
L_0=\tfrac12(1-t) K'_0(1), \qquad L_n=-\tfrac12(1-t)K'_n(1).
\end{align}
The new parameters $a,b,c$ and $d$ are related to the hopping rates $\alpha, \beta, \gamma$ and $\delta$ through the formul\ae
 \begin{align}
&\alpha =\frac{-a c (1-t)}{(1+a) (1+c)}, \qquad \gamma =\frac{1-t}{(1+a) (1+c)},\\
&\beta =\frac{-b d (1-t)}{(1+b) (1+d)},  \qquad \delta =\frac{1-t}{(1+b) (1+d)}.
\end{align}

In order to be able to define a family of commuting transfer matrices, we associate for later purposes the dual K-matrix to each K-matrix (\ref{K01}) and (\ref{Kn1}) in the following way. Introduce the matrix $\widetilde{R}(x)$
\begin{align}\label{Rwidetilde}
\widetilde{R}(x)= \left(\left(R\left(x\right)^{\tau_1}\right){}^{-1}\right){}^{\tau _1},
\end{align}
where the superscript $\tau_1$ denotes the transposition $\tau$ of the first component of the tensor product $(a\otimes b)^{\tau_1}=a^{\tau}\otimes b$, then the dual matrices $\widetilde{K}_n$ and $\widetilde{K}_0$
\begin{align}
&\widetilde{K}_0(x)=\text{Tr}_2\left(  ( \mathbb{I}_3 \otimes K_0\left(x \right)) \widetilde{R}(x^2)P \right),\label{Kd01}\\
&\widetilde{K}_n(x)=\text{Tr}_1\left(  (K_n\left(x^{-1}\right)\otimes \mathbb{I}_3) \widetilde{R}(x^2)P \right), \label{Kdn1}
\end{align}
where $\text{Tr}_1$ and $\text{Tr}_2$ are traces taken over the first and second components of the tensor product respectively. The resulting K-matrices read
\begin{align}
&\widetilde{K}_0(x)=\kappa_0(x) \left( U^{-1} + \frac{x^2-t^3}{t^2 h_0\left(a,c,x/t\right)}
\begin{pmatrix}
-t_0 & 0 & 1  \\
0 & 0 & 0 \\
t_0/t^{-1} & 0 & -1/t
\end{pmatrix}\right),\\
&\widetilde{K}_n(x)=\kappa_n(x)\left(U - \frac{x^2-t^3}{t x^2 h_n\left(b,d,t /x\right)}
\begin{pmatrix}
1 & 0 & -t_n  \\
0 & 0 & 0 \\
-t & 0 &t~ t_n
\end{pmatrix}\right), \\
&\kappa_0(x)= \frac{t^2 \left(x^2-t\right) h_0\left(a,c,x/t\right)}{\left(x^2-t^3\right) h_0(a,c,x)}, \qquad
\kappa_n(x)=\frac{\left(x^2-t\right) h_n\left(b,d,t/x\right)}{\left(x^2-t^3\right) h_n\left(b,d,1/x\right)}, \nonumber
\end{align}
where $U=\text{diag}\{t^{-1},1,t\}$.

\subsubsection{Commuting transfer matrices.}
\label{se:Tmat}
The R-matrix and the (dual) boundary K-matrices are used to build the transfer matrix $T(x)$ who's logarithmic derivative evaluated at $x=1$ yields the transition matrix $L$. The transfer matrices at different values of the parameter $x$ commute with each other due to the intertwining relations and unitarity conditions satisfied by the R-matrix and the K-matrices. 

The transfer matrix for a system with $n$ sites is an operator acting in the space $\mathcal{H}_{n,r}=V^{(r)}_1\otimes V^{(r)}_2 \otimes \cdots \otimes V^{(r)}_n$, a tensor product of $n$ copies of $V_i^{(r)}\simeq\mathbb{C}^{2r+1}$. This matrix is built by multiplying a two-row monodromy matrix $M(w)$ by the dual K-matrix $\widetilde{K}_n$ and tracing over an auxiliary space denoted by $V_0^{(r)}$. First we define two single row monodromy matrices $M^{(1)}(w)$ and $M^{(2)}(w)$ which are constructed as follows
\begin{align}
&M^{(1)}(w)=R_{0,n}(w) R_{0,n-1}(w)\dots R_{0,1}(w), \label{Mon1}\\
&M^{(2)}(w)=R_{1,0}(w) R_{2,0}(w)\dots R_{n,0}(w),  \label{Mon2}.
\end{align}
Using the K-matrix $K_0$, the  two row monodromy matrix $M(w)$ is then defined by
\begin{align}
&M(w)=M^{(1)}(w) K_0(w) M^{(2)}(w).\label{Mon}
\end{align}
Considering $M(w)$ as a $3\times 3$ matrix in $V_0$ we finally arrive at the definition of the transfer matrix $T(w)$:
\begin{align}\label{Tmat}
T(w)=\text{Tr}_0 ( M(w) \widetilde{K}_n (w) ).
\end{align}

A more general version of the transfer matrix is obtained if we associate with each space $V_j^{(r)}$ in $\mathcal{H}_{n,r}$ an inhomogeneity parameter $x_j$. In this case we must replace the R-matrices $R_{0,j}(w)$ in (\ref{Mon1}) by $R_{0,j}(w x_j)$ and $R_{j,0}(w)$ in (\ref{Mon2}) by $R_{j,0}(w/x_j)$. This leads to the inhomogeneous monodromy matrices $M^{(1)}(w;x_1,..,x_n)$ and $M^{(2)}(w;x_1,..,x_n)$ and to the inhomogeneous transfer matrix $T(w;x_1,..,x_n)$. We abbreviate $T(w)=T(w;1,..,1)$.

The following is a standard result in the theory of integrable models.
\begin{thm}
\label{th:commut}
The transfer matrix $T(w)$ defined in \eqref{Tmat} form a commuting family, i.e. $[T(w_1),T(w_2)]=0$. Furthermore, the Markov matrix $L$ is obtained form $T(w)$ define in \eqref{Tmat} by taking the derivative at $w=1$. In our conventions $T(1;x_1,\ldots,x_n))=\mathbb{I}$ (as well as $K_0(1)=K_n'(1)=\mathbb{I}_3$, $R_i(1)=P_i$), hence we simply get
\begin{align}
\label{eq:LderT}
L=\frac{1-t}{2}T'(1)=(1-t)
\left(\frac{1}{2} K_0'(1)+\sum_{i=1}^{n-1}\check{R}_{i,i+1}'(1)-\frac{1}{2} K_n'(1) \right).
\end{align}
\end{thm}

\begin{proof}
It is a standard calculation to prove that  $[T(w_1),T(w_2)]=0$, for any $w_1$ and $w_2$ \cite{Skl}. In order to prove it one needs to use the unitarity relation, Yang--Baxter equation, crossing unitarity relation and the reflection equations which we give below. We give the details of this calculation in \ref{tmcomm}. The commutativity of the transfer matrices implies $[T(w),L]=0$, therefore the matrices $T(w)$ and $L$ share the same eigenvectors. We list below the various relations needed for commutativity to hold.

To prove \eqref{eq:LderT} let us define the $n$ matrices $S_i(x_1,..,x_n)=T(w=x_i;x_1,..,x_n)$ for $i=1,..,n$. Setting the parameter $w=x_i$ reduces the transfer matrix to the following form
\begin{align}\label{Smat}
S_i(x_1,..,x_n)=& R_{i,i-1}(x_i/x_{i-1})..R_{i,1}(x_i/x_{1}) K_0(x_i) R_{1,i}(x_i x_1)..R_{i-1,i}(x_i x_{i-1}) \times \nonumber \\
			& R_{i+1,i}(x_i x_{i+1})..R_{L,i}(x_i x_L) K_n(1/x_i) R_{i,L}(x_i/x_{L})..R_{i,i+1}(x_i/x_{i+1}).
\end{align}
The derivative at the homogeneous point $x_1=x_2=\ldots =x_n=1$ is equivalent to
\begin{align}
\left. \frac{\dd}{\dd x_i}\right|_{\{x_j=1\}} S_i(x_1,..,x_n) = \left. \frac{\dd}{\dd w}\right|_{w=1} T(w) +  \left. \frac{\dd}{\dd x_i}\right|_{x_i=1} T(1;x_i) = T'(1),
\end{align}
because the second derivate vanishes as $T(1;x_1,\ldots,x_n)=\mathbb{I}$. To prove \eqref{eq:LderT} it is now a straightforward calculation to show that
\[
L= \frac12 (1-t) \left. \frac{\dd}{\dd x_i}\right|_{\{x_j=1\}} S_i(x_1,..,x_n).
\]
We list below the key ingredients for Theorem~\eqref{th:commut} to hold.

\bigskip
\noindent$\bullet$ \textit{Yang-Baxter equation}\\
The first intertwining relation is the Yang--Baxter equation written in $\mathcal{H}_{3,1}$
\begin{align}\label{YBeq}
R_{1,2}\left(y/x\right)R_{1,3}\left(y/z\right)R_{2,3}\left(x/z\right)=R_{2,3}\left(x/z\right)R_{1,3}\left(y/z\right) R_{1,2}\left(y/x\right).
\end{align}
While the unchecked R-matrices are most natural to define the transfer matrix, many of the fundamental relations are more naturally written using the checked matrices,in particular with a view to the Hecke algebra formulation in Section~\ref{se:Hecke}. Multiplying by $P_{1,2}P_{1,3}P_{2,3}$ both sides of this equation from the left  we get the Yang--Baxter equation in terms of $\check{R}$,
\begin{align}\label{YBcheck}
\check{R}_{2,3}\left(y/x\right)\check{R}_{1,2}\left(y/z\right)\check{R}_{2,3}\left(x/z\right)=\check{R}_{1,2}\left(x/z\right)\check{R}_{2,3}\left(y/z\right)
   \check{R}_{1,2}\left(y/x\right),
\end{align}
which is a version of the braid relation \eref{eq:braidi} that we shall encounter in Section~\ref{se:Hecke}.\\

\noindent$\bullet$ \textit{Unitarity and Crossing}\\
The matrix $R(x)$ satisfies the unitarity condition
\begin{align}\label{Unitarity}
R_{1,2}\left(x\right)R_{2,1}\left(1/x\right)=\mathbb{I}_3 \otimes \mathbb{I}_3 = \check{R}(x)\check{R}(1/x),
\end{align}
and the crossing unitarity condition
\begin{align}\label{CUnitarity}
\quad (\mathbb{I}_3\otimes U)R_{2,1}\left(t^3/x\right)^{\tau_1}(\mathbb{I}_3\otimes U^{-1})R_{1,2}(x)^{\tau_1}=-
\frac{(x-1) \left(t^3-x\right)}{(t-x)\left(t^2-x\right)}   \mathbb{I}_3\otimes\mathbb{I}_3.
\end{align}

\noindent$\bullet$ \textit{Left hand side reflection equation}\\
The boundary matrix $K_0(x)$ satisfies the reflection equation written in $\mathcal{H}_{2,1}$,
\begin{multline}
\label{ReflK0Check}
\left(K_0(x) \otimes \mathbb{I}_3 \right)\check{R}(wx)\left(K_0(w)\otimes\mathbb{I}_3\right) \check{R}(w/x)= \\
  \check{R}(w/x)\left(K_0(w)\otimes \mathbb{I}_3\right) \check{R} (wx)\left(K_0(x)\otimes \mathbb{I}_3 \right),
\end{multline}
which should be compared to the boundary braid relation \eref{eq:braid0} in the Hecke algebra. In terms of the unchecked $R$ this equation can also be written as

\begin{multline}\label{ReflK0}
\left(\mathbb{I}_3\otimes K_0(x)\right) R_{1,2}(w x)\left(K_0(w)\otimes\mathbb{I}_3\right)R_{2,1}(w/x)=\\
  R_{1,2}(w/x)\left(K_0(w)\otimes \mathbb{I}_3\right) R_{2,1} (w x)\left(\mathbb{I}_3\otimes K_0(x)\right).
\end{multline}
The reflection equations with the dual K-matrix $\widetilde{K}_0$ reads
\begin{multline}
\left (\mathbb{I}_3\otimes \widetilde{K}_0(x)\right) \widetilde{R}_{1,2}(w x) \left(\widetilde{K}_0 (w)\otimes\mathbb{I}_3\right) R_{1,2}(x/w)= \\ R_{2,1}(x/w)\left(\widetilde{K}_0(w)\otimes \mathbb{I}_3\right)\widetilde{R}_{2,1}(w x)\left(\mathbb{I}_3\otimes\widetilde{K}_0(x)\right).\label{ReflDK0}
\end{multline}

\noindent$\bullet$ \textit{Right hand side reflection equation}\\
Similarly for the K-matrix $K_n(x)$ we have
\begin{multline}\label{ReflKnCheck}
\left(\mathbb{I}_3\otimes K_n(x)\right) \check{R}(1/w x)\left(\mathbb{I}_3\otimes K_n(w)\right) \check{R}(x/w)= \\
\check{R}(x/w)\left(\mathbb{I}_3\otimes K_n(w) \right) \check{R}(1/w x)\left(\mathbb{I}_3\otimes K_n(x)\right),
\end{multline}
which should be compared to \eref{eq:braidn}, and in unchecked form is written as
\begin{multline}\label{ReflKn}
\left(\mathbb{I}_3\otimes K_n(x)\right) R_{2,1}(1/w x)\left(K_n(w)\otimes \mathbb{I}_3\right) R_{1,2}(x/w)= \\
R_{2,1}(x/w)\left(K_n(w)\otimes \mathbb{I}_3\right)R_{1,2}(1/w x)\left(\mathbb{I}_3\otimes K_n(x)\right).
\end{multline}
The reflection equation with the dual K-matrix $\widetilde{K}_n$ reads
\begin{multline}
\left(\mathbb{I}_3\otimes\widetilde{K}_n(x)\right) \widetilde{R}_{2,1}(w x) \left(\widetilde{K}_n(w)\otimes\mathbb{I}_3\right) R_{2,1}(x/w)= \\
R_{1,2}(x/w)\left(\widetilde{K}_n(w)\otimes \mathbb{I}_3\right)\widetilde{R}_{2,1}(w x)\left(\mathbb{I}_3\otimes\widetilde{K}_n(x)\right). \label{ReflDKn}
\end{multline}
\end{proof}

We also note that one can rewrite the transfer matrix \eqref{Tmat} in a form where the K-matrix $K_0$ is replaced by its dual and the dual matrix $\widetilde{K}_n (w)$ is replaced by  $K_n(w)$. This is possible because of the form of the dual matrices, containing  $\widetilde{R}(w^2)$ which intertwines the neighbouring R-matrices in $T(w)$. Explicitly this means that we can write $T$ as
\begin{align}
&\overline{M}(w;x_1,..,x_n)=M^{(2)}(w;x_1,..,x_n) K_n(1/w) M^{(1)}(w;x_1,..,x_n), \nonumber \\
&T(w;x_1,..,x_n)=\text{Tr}_0 ( \overline{M}(w;x_1,..,x_n) \widetilde{K}_0 (w) ). \label{Tmatd}
\end{align}
More details in the case of $n=2$ are provided in \ref{se:n=2}

\subsubsection{Exchange relations.}
Finally, let us mention several properties of the inhomogeneous transfer matrix. As a result of (\ref{YBcheck}) we have the bulk exchange relation
\begin{align}\label{RT}
\check{R}_{i}(x_{i+1}/x_i)T(w;x_1,..,x_i,x_{i+1},..,x_n)=T(w;x_1,..,x_{i+1},x_{i},..,x_n)\check{R}_{i}(x_{i+1}/x_i).
\end{align}
Equation (\ref{ReflK0}) and the inhomogeneous version of the definition (\ref{Tmat}) of the transfer matrix lead to the first boundary exchange relation
\begin{align}\label{K0T}
K_{0}(x_1)T(w;x_1,..,x_n)=T(w;1/x_1,..,x_n)K_{0}(x_1).
\end{align}
The second boundary exchange relation 
\begin{align}\label{KnT}
K_{n}(x_n)T(w;x_1,..,x_n)=T(w;x_1,..,1/x_n)K_{n}(x_n),
\end{align}
is satisfied due to (\ref{ReflKn}).

\section{Multi-species or higher rank ASEP}

The state space $\mathcal{H}_{n,r}:= \left(\mathbb{C}^{2r+1}\right)^{\otimes n}$ of the multi-species asymmetric exclusion process is spanned by the standard basis
\be
\ket{\mu} = \ket{\mu_1,\ldots,\mu_n},\qquad \mu_i\in\{-r,\ldots,r\}.
\ee

The $2r+1$-species asymmetric exclusion process is based on the R-matrix of $U_{t^{1/2}}(A^{(1)}_r)$, which can be expressed in the form
\begin{align}
\label{Rr}
\check{R}^{(2r+1)} (x)
&=
\sum_{i=-r}^{r}
E^{(ii)}
\otimes
E^{(ii)}
+
\sum_{-r \leq i < j \leq r}
\Big(
b^+ E^{(ij)}
\otimes
E^{(ji)}
+
b^- 
E^{(ji)}
\otimes
E^{(ij)}
\Big)\nonumber
\\
&\sum_{-r \leq i < j \leq r}
\Big(
c^-
E^{(ii)}
\otimes
E^{(jj)}
+
c^+
E^{(jj)}
\otimes
E^{(ii)}
\Big)\\
&=
\sum_{i=-r}^{r}
E^{(ii)}
\otimes
E^{(ii)}
+
\frac{1- x}{t - x}
\sum_{-r \leq i < j \leq r}
\Big(
t E^{(ij)}
\otimes
E^{(ji)}
+
E^{(ji)}
\otimes
E^{(ij)}
\Big)\nonumber
\\
& \frac{t-1}{t - x}
\sum_{-r \leq i < j \leq r}
\Big(
x
E^{(ii)}
\otimes
E^{(jj)}
+
E^{(jj)}
\otimes
E^{(ii)}
\Big). \nonumber
\end{align}
where $E^{(ij)}$ denotes the elementary $(2r+1) \times (2r+1)$ matrix with a single non-zero entry 1 at position $(i,j)$. The corresponding boundary matrix $K_0(x)$ is given by
\begin{align}
\label{K0r}
K^{(2r+1)}_0 (x)
=
\sum_{i=-r}^{r}
E^{(ii)}
+
\frac{1-x^2}{h_0(a,c,x)}
\Big(
\sum_{0 < i \leq r}
t_0 E^{(-i,-i)}
+
E^{(r+1-i,r+1-i)}
\nonumber\\
{}
-
\sum_{0 < i \leq r}
E^{(-i,r+1-i)}
+ t_0
E^{(r+1-i,-i)}
\Big),
\end{align}
and $K_n(x)$ is defined as
\begin{align}
\label{Knr}
K^{(2r+1)}_n (x)
=
\sum_{i=-r}^{r}
E^{(ii)}
-
\frac{1-x^2}{h_n(b,d,x)}
\Big(
\sum_{0 < i \leq r}
E^{(-i,-i)}
+t_n
E^{(r+1-i,r+1-i)}
\nonumber\\
{}
-
\sum_{0 < i \leq r}
E^{(r+1-i,-i)}
+t_n
E^{(-i,r+1-i)}
\Big).
\end{align}

The dual K-matrices read 
\begin{align}
&\widetilde{K}^{(2r+1)}_0 (x)
=\kappa_0^{(2r+1)}(x) \bigg{(}
\sum_{i=-r}^{0}
E^{(ii)} t^{-i}
+
\sum_{i=1}^{r}
E^{(ii)} t^{-r+i-1}
+
\frac{t^{-2 r} \left(x^2-t^{2 r+1}\right)}{h_0\left(a,c,x/t^{r}\right)} \times \nonumber \\
&\Big(
-\sum_{0 < i \leq r}
t_0 t^{i-1} E^{(-i,-i)}
+
t^{-i}E^{(r+1-i,r+1-i)}
+
\sum_{0 < i \leq r}
E^{(-i,r+1-i)}
+t^{-1} t_0
E^{(r+1-i,-i)}
\Big)\bigg{)},
\end{align}
and 
\begin{align}
&\widetilde{K}^{(2r+1)}_n (x)
=\kappa_n^{(2r+1)}(x) \bigg{(}
\sum_{i=-r}^{0}
E^{(ii)} t^{i}
+
\sum_{i=1}^{r}
E^{(ii)} t^{r-i+1}
-
\frac{t^{-r} \left(x^2-t^{2 r+1}\right)}{x^2 h_n\left(b,d,t^r/x\right)} \times \nonumber \\
&\Big(
\sum_{0 < i \leq r}
t^{r-i} E^{(-i,-i)}
+t^{r+i-1}t_n
E^{(r+1-i,r+1-i)}
-
\sum_{0 < i \leq r}
t^{r-1}t_n E^{(r+1-i,-i)}
+t^r
E^{(-i,r+1-i)}
\Big)\bigg{)},
\end{align}
where $\kappa_0^{(2r+1)}$ and $\kappa_n^{(2r+1)}$
\begin{align*}
\kappa_0^{(2r+1)}(x)=\frac{t^{2 r} \left(t-x^2\right) h_0\left(a,c,x/t^{r}\right)}{\left(t^{2 r+1}-x^2\right) h_0(a,c,x)},\qquad 
 \kappa_n^{(2r+1)}(x)=\frac{\left(t-x^2\right)
   h_n\left(b,d,t^r/x\right)}{\left(t^{2
   r+1}-x^2\right) h_n\left(b,d,1/x\right)}.
\end{align*}
Among the relations satisfied by the matrices $R^{(2r+1)}$, $K_0^{(2r+1)}$, and $K_n^{(2r+1)}$ the crossing unitarity condition changes to
\begin{align}\label{CUnitarityr}
(\mathbb{I}_{2r+1}\otimes U^{-1})R^{(2r+1)}_{2,1}\left(t^{2r+1}/x\right)^{\tau_1}(\mathbb{I}_3\otimes U)R_{1,2}^{(2r+1)}(x)^{\tau_1}=
\frac{(1-x) \left(t^{2r+1}-x\right)}{(t-x)\left(t^{2r}-x\right)}   \mathbb{I}_{2r+1}\otimes\mathbb{I}_{2r+1},
\end{align}
with $U=\text{diag}\{t^{-r},..,t^{-1},1,t,..,t^{r}\}$, while other relations remain unchanged as in previous subsection. The transfer matrix $T^{(2r+1)}$ is constructed in the same way as before according to (\ref{Mon1})--(\ref{Tmat}) and (\ref{Tmatd}). The $2r+1$-species ASEP Markov matrix is given by the derivative of the transfer matrix $T^{(2r+1)}$
\begin{align*}
L^{(2r+1)}=\frac{1-t}{2}T^{(2r+1)}{}'(1)=(1-t)\left(\frac{1}{2} K_0^{(2r+1)}{}'(1)+\sum_{i=1}^{n-1}\check{R}_{i,i+1}^{(2r+1)}{}'(1)-\frac{1}{2} K_n^{(2r+1)}{}'(1)\right).
\end{align*}
The exchange relations (\ref{RT})--(\ref{KnT}) also hold for the matrices $R^{(2r+1)}$, $K_0^{(2r+1)}$, and $K_n^{(2r+1)}$, respectively. Likewise the scattering matrices $S_i^{(2r+1)}$ are defined through  $T^{(2r+1)}(x)$.

\section{Relation to the Hecke algebra}
\label{se:Hecke}

\subsection{Weyl group}
\label{subsc:weyl}
The Weyl group of \textit{finite} type C 
\be
W_0 = \langle s_1,\ldots,s_n\rangle
\ee
acts on $\mathbb{R}^n$ as
\be
s_i (\lambda_1,\ldots,\lambda_n) =  (\lambda_1,\ldots\lambda_{i+1},\lambda_i,\ldots,\lambda_n),\qquad s_n (\lambda_1,\ldots,\lambda_n) =  (\lambda_1,\ldots,-\lambda_n).
\ee
The Weyl group $W_0$ is isomorphic to the group of signed permutations on $n$ symbols. Let $\lambda\in\mathbb{Z}^n$ be a composition and $w_\lambda\in W_0$ the shortest word so that $w_\lambda^{-1}\lambda=:\delta^+$ is a partition, the unique dominant weight. The word $w_\lambda\in W_0$ can be described as a signed permutation, $w_\lambda = \sigma_\lambda \pi_\lambda$, where $\sigma_\lambda=(\sign(\lambda_1),\ldots,\sign(\lambda_n))$ with $\sign(0)=1$, and $\pi_\lambda\in S_n$ is a permutation. 

The dominance order $\ge$ on $\mathbb{Z}^n$ is defined as
\be
\lambda\ge \mu\text{ if } \sum_{i=1}^k (\lambda_i - \mu_i) \ge 0\quad \text{for }k=1,\ldots,n,
\ee
and a partial order $\succeq$ is defined as
\be
\lambda \succeq \mu \text{ if } \lambda^+ >  \mu^+ \text{ or if } \lambda^+=\mu^+ \text{ and } \lambda\ge \mu.
\ee

The \textit{affine} Weyl group includes the generator $s_0$,
\be
W = \langle s_0,\ldots,s_n\rangle,
\ee
and has a natural faithful action on $\mathbb{R}^n$ where
\be
s_0  (\lambda_1,\ldots,\lambda_n) =  (-1-\lambda_1,\ldots,\lambda_n).
\ee

The simple transposition$s_i$ has a simple action on polynomials given by $s_i f(\ldots,x_i,x_{i+1},\ldots) = f(\ldots,x_{i+1},x_i,\ldots)$. Let also
\be
s_0f(x_1,\ldots)=f(q/x_1,\ldots),\qquad s_nf(\ldots,x_n)=f(\ldots,1/x_n).
\ee

\subsection{Exchange relations}
Fix a partition $\lambda$, then 
\begin{align}
\label{thetadef}
\bra{\theta_\lambda} := \sum_{\mu \in W_0\cdot\lambda } \bra{\mu}
\end{align}
is the Perron-Frobenius left eigenvector\footnote{The vector $\bra{\theta_\lambda}$ is just the row vector $(1,1,1,\ldots)$ on the subspace spanned by signed permutations of $\lambda$.} of the discrete time transfer matrix defined in Section~\ref{se:Tmat},
\be
\bra{\theta_\lambda} T(w;x_1,\ldots,x_n) = \Lambda_\lambda(w;x_1,\ldots,x_n)  \bra{\theta_\lambda} .
\ee
This result can be easily proved by observing that $\bra{\theta_\lambda}$ is a left eigenvector of $S_i$ (with eigenvalue 1), and that by commutativity of $S_i$ and $T$ it also has to be an eigenvector of $ T(w;x_1,\ldots,x_n)$. We define the corresponding Perron-Frobenius right eigenstate as
\begin{align}
\ket{\Psi_\lambda(x_1,\ldots,x_n)} &= \sum_{\mu \in W_0\cdot\lambda } f_{\mu_1,\ldots,\mu_n}(x_1,\ldots,x_n) \ket{\mu},\\
T(w;x_1,\ldots,x_n) \ket{\Psi_\lambda(x_1,\ldots,x_n)} &= \Lambda_\lambda(w;x_1,\ldots,x_n)  \ket{\Psi_\lambda(x_1,\ldots,x_n)}.
\end{align}

Let $\check{R}$ and $K_n$ be as define in \eqref{Rr} and \eqref{Knr}, and define a $q$-deformed modification of $K_0$ in \eqref{K0r} as 
\begin{align}
K^{(2r+1)}_0 (x)
=
\sum_{i=-r}^{r}
E^{(ii)}
+
\frac{q-x^2}{h_0(a,c,x)}
\Big(
\sum_{0 < i \leq r}
t_0 E^{(-i,-i)}
+
E^{(r+1-i,r+1-i)}
\nonumber\\
{}
-
\sum_{0 < i \leq r}
q^{-i}
E^{(-i,r+1-i)}
+q^{r+1-i} t_0
E^{(r+1-i,-i)}
\Big).
\end{align}
\begin{lemma}[Cantini, \cite{Cantini15}]
The inhomogeneous Perron-Frobenius eigenstate 
\be
\ket{\Psi(x_1,\ldots,x_n)}  = \sum_\mu f_{\mu_1,\ldots,\mu_n}(x_1,\ldots,x_n) \ket{\mu},
\ee
 is the $q=1$ solution of 
\begin{align}
\check{R}_i(x_{i+1}/x_i)\ \ket{\Psi(x_1,\ldots,x_n)} & =s_i\ket{\Psi(x_1,\ldots,x_n)}, \nonumber\\
K_0(x_1)\ \ket{\Psi(x_1,\ldots,x_n)} & =s_0 \ket{\Psi(x_1,\ldots,x_n)} ,
\label{eq:qKZ} \\
K_n(x_n)\ \ket{\Psi(x_1,\ldots,x_n)} & =s_n \ket{\Psi(x_1,\ldots,x_n)}. \nonumber
\end{align}
\label{eq:qKZlemma}
\end{lemma}

\begin{proof}
This follows from the exchange relations \eqref{RT}, \eqref{K0T} and \eqref{KnT} of $\check{R}$, $K_0$ and $K_n$ with the transfer matrix. 
\end{proof}

In the following we describe $\ket{\Psi(x_1,\ldots,x_n)}$ using the polynomial representation theory of the Hecke algebra.

\subsection{Hecke algebra of affine type C}
The Weyl group $W$ can be $t$-deformed to the Hecke algebra of affine type C. In the polynomial representation \cite{Noumi} the generators of the Hecke algebra are explicitly given by
\begin{align}
T_0 &= t_0  -  \frac{(x_1 -a)(x_1-c)}{x_1^2-q} (1-s_0),\nonumber \\
T_i  &= t - \frac{t x_i - x_{i+1}}{x_i-x_{i+1}} (1-s_i) \qquad (i=1,\ldots,n-1),\label{eq:Hecke}\\
T_n &= t_n  -  \frac{(bx_n -1)(dx_n-1)}{1-x_n^2} (1-s_n).\nonumber
\end{align}
where 
\be
t_0=-a cq^{-1},\qquad t_n=-b d.
\ee
It can be checked straightforwardly that the braid relations
\begin{align}
T_0T_1T_0T_1 &= T_1T_0T_1T_0,\label{eq:braid0}\\
T_iT_{i+1} T_i &= T_{i+1}T_iT_{i+1},\label{eq:braidi}\\
T_nT_{n-1}T_nT_{n-1} &= T_{n-1}T_nT_{n-1}T_n,\label{eq:braidn}
\end{align}
are satisfied  and so are the quadratic relations
\be
(T_i-t)(T_i+1) = 0, \quad (T_0-t_0)(T_0+1) = 0, \quad (T_n-t_n)(T_n+1) = 0. 
\label{eq:Hecke2}
\ee

It is convenient to define the following shifted operator (sometimes referred to as Baxterised operator),
\begin{equation}
T_i(u)  = T_i +\frac{1}{[u]},\qquad  [u]= \frac{1- t^{u}}{1-t}.
\label{eq:Demazure} 
\end{equation}
which satisfies the Yang--Baxter equation,
\begin{equation}
T_i(u)T_{i+1}(u+v)T_i(v) = T_{i+1}(v) T_i(u+v) T_{i+1}(u).
\end{equation}

\subsection{Non-symmetric Koornwinder polynomials}
The operators $Y_i$ defined by \cite{Sahi}
\be
Y_i = (T_i \ldots T_{n-1} )(T_n\cdots T_{0})(  T_{1}^{-1} \cdots T_{i-1}^{-1}).\quad (i=1,\ldots,n),
\ee
form an Abelian subalgebra, and symmetric functions of these operators are central elements of the Hecke algebra. The set of operators $Y_i$ therefore share a common set of eigenfunctions and in the polynomial representation these eigenfunctions are non-symmetric Koornwinder polynomials.

Following Kasatani \cite{Kasatani}, solutions to the exchange equations in Lemma~\ref{eq:qKZlemma} can be obtained from the anti-dominant non-symmetric Koornwinder polynomial in the following way. Let $\lambda\in\mathbb{Z}^n$ be a composition. Let $\delta$ be the antidominant weight of $\lambda$ in the partial order $\succeq$ on compositions, i.e. $\delta$ is a signed permutation of $\lambda$ such that $\delta_1 \leq \delta_2 \leq \ldots \leq \delta_n \le 0$. Let furthermore $\rho(\delta)= w_+ \rho$, $\rho=(n-1,n-2,\ldots,0)$ and $w_+$ is the shortest word in $W_0$ such that $\delta = w_+ \delta^+$ where $\delta^+$ is the dominant weight. 

\begin{definition}
The non-symmetric Koornwinder polynomial $E_\lambda$ is the unique polynomial which solves the eigenvalue equations
\begin{equation}
Y_i E_\lambda = y_i(\lambda) E_\lambda \qquad (i=1,\ldots,n),
\end{equation}
where
\be
y_i(\lambda) =  q^{\lambda_i} t^{n-i+\rho(\lambda)_i} (t_0t_n)^{\epsilon_i(\lambda)}  ,\qquad \epsilon_i(\lambda)= \left\{ \begin{array} {lc} 1 & \lambda_i \ge 0 \\ 0 & \lambda_i <0 \end{array}\right.
\ee
and whose coefficient of the term $x^\lambda=x_1^{\lambda_1}\cdots x_n^{\lambda_n}$ is equal to 1.
\end{definition}

\begin{definition}
Denote by $\mathcal{R}=\mathbb{F}[x_1^{\pm},\ldots,x_n^{\pm 1}]$ the ring of Laurent polynomials in $n$ variables. The space $\mathcal{R}^\lambda$ is the subspace of $\mathcal{R}$ spanned by $\{E_\mu| \mu\in W_0\cdot\lambda\}$.
\end{definition}

Define
\begin{align}
f_\delta &:= E_\delta,\nonumber \\
f_{\ldots,\lambda_i,\lambda_{i+1},\ldots} &:=  T_i^{-1} f_{\ldots,\lambda_{i+1},\lambda_{i},\ldots} \qquad \lambda_i > \lambda_{i+1},
\label{qKZ1} \\
f_{\ldots,\lambda_{n-1},\lambda_n} &:= T_n^{-1} f_{\ldots,\lambda_{n-1},-\lambda_{n}}, \qquad \lambda_n > 0.
\nonumber
\end{align}
Then $f$ solves the following equations \cite{Kasatani},
\begin{align}
T_0 f_{\lambda_1,\ldots} &= q^{\lambda_1} f_{-\lambda_1,\ldots} \qquad \lambda_1 < 0,\nonumber
\\
T_0 f_{\lambda_1,\ldots} &= t_0 f_{\lambda_1,\ldots} \quad\qquad \lambda_1 =0,\nonumber
\\
T_i f_{\ldots,\lambda_i,\lambda_{i+1},\ldots} &= t f_{\ldots,\lambda_i,\lambda_{i+1},\ldots}\qquad \lambda_i = \lambda_{i+1},\nonumber
\\
T_i f_{\ldots,\lambda_i,\lambda_{i+1},\ldots} &=  f_{\ldots,\lambda_{i+1},\lambda_i,\ldots}\qquad \lambda_i > \lambda_{i+1},\label{qKZC1b}\\
T_n f_{\ldots,\lambda_n} &= t_n f_{\ldots,\lambda_n},\qquad \lambda_n=0,\nonumber
\\
T_n f_{\ldots,\lambda_{n}} &= f_{\ldots,-\lambda_n},\qquad \lambda_n > 0.\nonumber
\end{align}

\begin{lemma}
Equations \eqref{qKZC1b} are equivalent to \eqref{eq:qKZ}.
\end{lemma}
\begin{proof}
This lemma follows by a straightforward check.
\end{proof}

\begin{remark}
The first of equations \eqref{qKZC1b} for the case $$\lambda=\delta=((-n+1)^{d_{n-1}},\ldots,(-1)^{d_1}0^{d_0}),$$ with $n=d_0+d_1+\ldots +d_{n-1}$, follows from
\begin{align}
T_0 E_\delta &= T_1^{-1} \ldots T_n^{-1} T_{n-1}^{-1} \ldots T_1^{-1} Y_1 E_\delta =  t^{n-1+\rho(\delta)_1} q^{\delta_1}  T_1^{-1} \ldots T_n^{-1} T_{n-1}^{-1} \ldots T_1^{-1} E_\delta  \nonumber\\
&=  t^{n-1+\rho(\delta)_1-d_{n-1}} q^{\delta_1}  T_1^{-1} \ldots T_n^{-1} E_{\delta_2,\ldots,\delta_n,\delta_1} \nonumber\\
 &= q^{\delta_1}  T_1^{-1} \ldots T_{n-1}^{-1} E_{\delta_2,\ldots,\delta_n,-\delta_1} = q^{.\delta_1} E_{-\delta_1,\delta_2,\ldots,\delta_n} ,
\end{align}
where we made use of the fact that $\rho(\delta)_1=-(n-1-d_{n-1})$.
\end{remark}

Writing $q=t^{u}$ and $t_0t_n=t^v$ we define the elements of the \textit{spectral vector}  $\langle \lambda\rangle$ of a composition $\lambda$  as,
\be
\langle \lambda\rangle_i = \rho_i(\lambda)+ u \lambda_i + v\epsilon_i(\lambda), \qquad y_i(\lambda) = t^{\langle \lambda\rangle_i }.
\ee 
The non-symmetric Macdonald polynomials are obtained from $E_\delta$ by the action of Baxterised operators:
\begin{align}
t E_{s_i\lambda} &=  T_i(\langle \lambda\rangle_{i+1} -\langle \lambda\rangle_{i} ) E_\lambda,\qquad \lambda_i < \lambda_{i+1},
\label{eq:TionE}\\
t_n E_{s_n\lambda} &= \left[T_n + \frac{1-t_n+t_n(1-t_0)y_n(\lambda)^{-1}}{t_0t_ny_n(\lambda)^{-2}-1} \right] E_\lambda,\qquad \lambda_n<0\label{eq:TnonE}
\end{align}

\begin{prop}
\label{prop:changeofbasis}
The families of polynomials $E_\mu$ and $f_\mu$ are related via an invertible triangular change of basis:
\begin{align}
\label{eq:triang}
E_{\lambda}
=
\sum_{\mu \leq \lambda}
c_{\lambda\mu}(q,t)
f_{\mu},
\qquad
f_{\lambda}
=
\sum_{\mu \leq \lambda}
d_{\lambda\mu}(q,t)
E_{\mu}
\end{align}
for suitable rational coefficients $c_{\lambda\mu}(q,t)$ and $d_{\lambda\mu}(q,t)$.
\end{prop}

\begin{proof}
This follows directly from \eqref{eq:TionE} and \eqref{eq:TnonE} and the definition of $T_i(u)$ in \eqref{eq:Demazure}, together with the definition of $f_\mu$ in \eqref{qKZ1}.
\end{proof}

\begin{cor}
The set of polynomials $\{f_\mu | \mu\in W_0\cdot \lambda\}$ form a basis in the ring $\mathcal{R}^\lambda$.
\end{cor}
\begin{proof}
Since the set $\{E_\mu| \mu\in W_0\cdot \lambda\}$ is a basis for $\mathcal{R}^\lambda$ \cite{Sahi}, the statement follows from Prop.~\ref{prop:changeofbasis}.
\end{proof}

\section{Symmetric Koornwinder polynomials}

We relate our results to symmetric Koornwinder polynomials \cite{Koornwinder92,Diejen}. 
\begin{lemma}
\label{le:symmetry}
Let $\lambda$ be a dominant composition, i.e. a partition. Then the sum 
\be
\mathcal{K}_\lambda(x_1,\ldots,x_n;q,t) = \sum_{\mu \in W_0\cdot \lambda} f_{\mu} (x_1,\ldots,x_n;q,t),
\ee 
 is $W_0$-invariant. Here the sum runs through all distinct elements in the $W_0$-orbit of $\lambda$.
\end{lemma}
\begin{proof}
We need to show that $T_i \mathcal{K}_\lambda = t \mathcal{K}_\lambda$ for all $i=1,\ldots,n-1$ and that $T_n \mathcal{K}_\lambda = t_n \mathcal{K}_\lambda$. From \eref{qKZ1} and \eref{eq:Hecke} we find for $\lambda_i < \lambda_{i+1}$ that
\begin{align}
T_i f_{\ldots,\lambda_i,\lambda_{i+1},\ldots} &= T_i^2  f_{\ldots,\lambda_{i+1},\lambda_{i},\ldots} = \left(t+(t-1)T_i \right)  f_{\ldots,\lambda_{i+1},\lambda_{i},\ldots} \nonumber\\
&= t f_{\ldots,\lambda_{i+1},\lambda_{i},\ldots} +(t-1)  f_{\ldots,\lambda_{i},\lambda_{i+1},\ldots}.
\end{align}
Combining this with \eref{eq:Hecke2} we thus find
\begin{align}
T_i \sum_\mu f_\mu &= \sum_{\mu:\ \mu_i < \mu_{i+1}} \left( tf_{s_i\mu} +(t-1)f_\mu \right) + \sum_{\mu:\ \mu_i = \mu_{i+1}} t f_\mu + \sum_{\mu:\ \mu_i > \mu_{i+1}} f_{s_i\mu} \nonumber\\
&= \sum_{\mu:\ \mu_i < \mu_{i+1}} tf_{s_i\mu} + \sum_{\mu:\ \mu_i \le \mu_{i+1}} tf_\mu = t \sum_\mu f_\mu.
\end{align}

Likewise, for $\lambda_n>0$
\begin{align}
T_n f_{\ldots,-\lambda_n} &= T_n^2  f_{\ldots,\lambda_{n}} = \left(t_n+(t_n-1)T_i \right)  f_{\ldots,\lambda_{n}} \nonumber\\
&= t_n f_{\ldots,\lambda_{n}} +(t_n-1)  f_{\ldots,-\lambda_{n}},
\end{align}
and therefore
\begin{align}
T_n \sum_\mu f_\mu &= \sum_{\mu:\ \mu_n < 0} \left( t_n f_{s_n\mu} +(t_n-1)f_\mu \right) + \sum_{\mu:\ \mu_n = 0} t_n f_\mu + \sum_{\mu:\ \mu_n > 0} f_{s_n \mu} \nonumber\\
&= \sum_{\mu:\ \mu_n < 0 } t_nf_{s_n\mu} + \sum_{\mu:\ \mu_n \le 0}  t_n f_\mu = t_n \sum_\mu f_\mu.
\end{align}
\end{proof}

The Koornwinder polynomial $K_{\lambda}$ is the unique $W$-symmetric polynomial (up to normalisation) which can be obtained by taking linear combinations of the non-symmetric Koornwinder polynomials $E_{\mu}$, where $\mu$ is a signed permutation of $\lambda$. As $\{f_\mu\}$ is a basis for $\mathcal{R}^\lambda$ it follows that $\mathcal{K}_{\lambda}=K_{\lambda}$.

\begin{thm}
\label{th:ZK}
The normalisation of the stationary state of the $2r+1$-species asymmetric exclusion process with open boundary conditions is a specialisation of a Koornwinder polynomial at $q=1$, i.e. 
\be
Z_{\lambda}(t,a,b,c,d) = K_\lambda(1^n;q=1,t;a,b,c,d).
\ee
\end{thm}

\begin{cor}
The normalisation of the stationary state of the $2r+1$-species asymmetric exclusion process factorises as a product over the rank $r=1$ standard ASEP 
\be
Z_{\lambda}(t,a,b,c,d) = \prod_i Z_{\lambda'_i}(t,a,b,c,d)
\ee
\end{cor}
\begin{proof}
It is a property \cite{RainsW} of Koornwinder polynomials that at $q=1$ we have
\be
K_{\lambda}(x_1,\ldots,x_n;1,t;a,b,c,d) = \prod_i K_{1^{\lambda'_i}}(x_1,\ldots,x_n;1,t;a,b,c,d),
\ee
where $1^k$ denotes a column of length $k$. The Corollary then follows immediately from Theorem~\ref{th:ZK}.
\end{proof}
\bigskip

\section{Generalised boundary conditions}
\label{se:genbc}

So far we have treated the label $0$ as special, as it cannot be
created nor annihilated at the boundaries. It is possible within our
setup to take similar boundary conditions for labels
$\{0,\ldots,r_L\}$ and take nonzero boundary rates for the following
events at the left hand side, 
\be
\begin{split}
(-m,\ldots) &\mapsto (+m,\ldots)\quad \text{with rate } \alpha,\\
(+m,\ldots) &\mapsto (-m,\ldots)\quad \text{with rate } \gamma,\label{eq:BCL-new}
\end{split}
\ee
for $m\in\{r_L+1,\ldots,r\}$ and likewise at the right hand boundary
\be
\begin{split}
(\ldots,+m) &\mapsto (\ldots,-m)\quad \text{with rate } \beta,\\
(\ldots,-m) &\mapsto (\ldots,+m)\quad \text{with rate } \delta.\label{eq:BCR-new}
\end{split}
\ee
for $m\in\{r_R+1,\ldots,r\}$.

Clearly the case dealt in the previous sections corresponds to
$r_L=r_R=0$. Notice that all the particles of species with label 
$|\mu|\leq \min(r_L,r_R)$ do not get flipped at either boundaries so
their number is conserved. Particles with label
$\min(r_L,r_R)<|\mu|\leq \max(r_L,r_R)$ can be flipped only at the
boundary corresponding to the minimum. These boundary
conditions are a sub famility of the boundary conditions considered by
Crampe et al. \cite{CrampeFRV}, moreover the case $r_L=r_R=r=1$ has been considered by
\cite{Schutz2015,Kuan2015}, while the case  $r_L=r_R=r$ has been
considered by \cite{Mandelshtam2015}. The corresponding
boundary matrix $K_0(x)$ is given by 
\begin{align}
K^{(2r+1,r_L)}_0 (x)
=
\sum_{i=-r}^{r}
E^{(ii)}
+
\frac{q-x^2}{h_0(a,c,x)}
\Big(
\sum_{r_L < i \leq r}
t_0 E^{(-i,-i)}
+
E^{(r+1-i,r+1-i)}
\nonumber\\
{}
-
\sum_{r_L < i \leq r}
q^{-i}
E^{(-i,r+1-i)}
+q^{r+1-i} t_0
E^{(r+1-i,-i)}
\Big),
\end{align}
and $K_n(x)$ is defined as
\begin{align}
K^{(2r+1,r_R)}_n (x)
=
\sum_{i=-r}^{r}
E^{(ii)}
-
\frac{1-x^2}{h_n(b,d,x)}
\Big(
\sum_{r_R < i \leq r}
E^{(-i,-i)}
+t_n
E^{(r+1-i,r+1-i)}
\nonumber\\
{}
-
\sum_{r_R < i \leq r}
E^{(r+1-i,-i)}
+t_n
E^{(-i,r+1-i)}
\Big).
\end{align}
Let us assume without loosing generality that $r_R\leq r_L$ (the
opposite case can be treated analogously). The  action of the Weyl group
$W_0$ on $\mathbb Z^n$ defined in Section \ref{subsc:weyl} can be
deformed by 
\be\begin{split}
s_i (\lambda_1,\ldots,\lambda_n) &=
(\lambda_1,\ldots\lambda_{i+1},\lambda_i,\ldots,\lambda_n),\\
s_n (\lambda_1,\ldots,\lambda_n) &= \left\{
\begin{array}{cc}
(\lambda_1,\ldots,-\lambda_n) & |\lambda_n|>r_R \\
(\lambda_1,\ldots,\lambda_n) & |\lambda_n|\leq r_R.
\end{array} \right.
\end{split}\ee
This action splits $\mathbb{Z}^n$ into sectors, that are labeled by
\emph{generalised} dominant 
weights $\bar \delta^+$, which are weakly decreasing compositions in
$\mathbb{Z}^n$, such that their entries are larger or equal to
$-r_R$. Let $\bar \delta$ the antidominant weight of $\bar
\delta^+$, i.e. $\bar \delta$ antidominant and  $\bar \delta\in
W_0(\bar \delta^+)$. 
For any composition $\mu \in \mathbb Z^n$ call $w_\mu$ the shortest
signed permutation that puts $\mu$ in antidominant form. By abuse of
notation we call $\ell(\mu)=\ell(w_\mu)$, the length of $w_\mu$, and 
$m(\mu)=\#\{(w_\mu)_i<0\}$, i.e. the number of minus signs in $w_\mu$.
Then to a composition $\mu$ we associate two compostions $\mu^c$
and $\mu^\pi$. The first one, $\mu^{c}$, is obtained from $\mu$ by removing all its entries whose
modulus is larger than $r_L$. The second one, namely $\mu^\pi$,
is defined by 
$$
\mu^\pi_i= \left\{
\begin{array}{cc}
0 & |\mu_i| \leq r_L\\
\mu_i & |\mu_i| > r_L
\end{array}
\right.
$$
Notice in particular that if $\bar \lambda$ is a generalised
dominant weight, then   $\bar \lambda^\pi$ is a dominant weight whose
parts are either zero or larger than $r_R$. 
Given a generalised dominant weight $\bar \lambda$, denote by
\be
\ket{\Psi^{\bar \lambda}(x_1,\ldots,x_n)}  = \sum_{\mu\in W_0 \bar
  \lambda} f^{\bar \lambda}_{\mu_1,\ldots,\mu_n}(x_1,\ldots,x_n) \ket{\mu},
\ee
a solution of the equations
\begin{align}
\check{R}_i(x_{i+1}/x_i)\  \ket{\Psi^{\bar \lambda}(x_1,\ldots,x_n)} & =s_i\ket{\Psi^{\bar \lambda}(x_1,\ldots,x_n)}, \nonumber\\
K_0(x_1)\ \ket{\Psi^{\bar \lambda}(x_1,\ldots,x_n)} & =s_0 \ket{\Psi^{\bar \lambda}(x_1,\ldots,x_n)} ,
\label{eq:qKZ-gen} \\
K_n(x_n)\ \ket{\Psi^{\bar \lambda}(x_1,\ldots,x_n)} & =s_n \ket{\Psi^{\bar \lambda}(x_1,\ldots,x_n)}. \nonumber
\end{align}
which in terms of the components $f^{\bar \lambda}_{\mu_1,\ldots,\mu_n}$ read 
\begin{align}
T_0 f^{\bar \lambda}_{\mu_1,\ldots} &= q^{\mu_1}
f^{\bar \lambda}_{-\mu_1,\ldots} \qquad \mu_1 <
-r_L,\nonumber
\\
T_0 f^{\bar \lambda}_{\mu_1,\ldots} &= t_0
f^{\bar \lambda}_{\mu_1,\ldots} \qquad |\lambda_1| \leq  r_L,\nonumber
\\
T_i f^{\bar \lambda}_{\ldots,\mu_i,\mu_{i+1},\ldots} &= t f^{\bar \lambda}_{\ldots,\mu_i,\mu_{i+1},\ldots}\qquad \mu_i = \mu_{i+1},\nonumber
\\
T_i f^{\bar \lambda}_{\ldots,\mu_i,\mu_{i+1},\ldots} &=  f^{\bar \lambda}_{\ldots,\mu_{i+1},\mu_i,\ldots}\qquad \mu_i > \mu_{i+1},\label{qKZC1b-s}\\
T_n f^{\bar \lambda}_{\ldots,\mu_n} &= t_n f^{\bar \lambda}_{\ldots,\mu_n}\qquad
|\mu_n|\leq r_R,\nonumber
\\
T_n f^{\bar \lambda}_{\ldots,\mu_{n}} &= f^{\bar \lambda}_{\ldots,-\mu_n}\qquad \mu_n > r_R.\nonumber
\end{align}
If the generalised dominant weight $\bar \lambda$ has parts strictly
larger than $r_L$ or zero, i.e. if $\bar \lambda = \bar \lambda^\pi$,
then these equations coincide with 
(\ref{qKZC1b}), therefore their solution is given by
(\ref{qKZ1}). The other extreme case is when all the parts of
$\bar \lambda$ are  in modulus smaller or equal to $r_L$, i.e. if
$\bar \lambda= \bar \lambda^c$, in this case it is easy to verify that
the solution of (\ref{qKZC1b-s})  is given by
\be
f^{\bar \lambda}_{\mu}(x_1,\dots,x_n)= t^{-\ell(\mu)}  (tt_n^{-1})^{m(\mu)},
\ee 
and since this does not depend on the spectral parameter we shall write
it in the following as $\ket{\Psi^{\bar  \lambda}}$. 
For the general case the solution has a nested form given by the following
\begin{prop}
Let $\bar \lambda$ be a generalised dominant weight, then the
solution of (\ref{qKZC1b-s}) is given by
\be
\ket{\Psi^{\bar \lambda}(x_1,\ldots,x_n)}  = \ket{\Psi^{\bar
    \lambda^\pi}(x_1,\ldots,x_n)} \otimes \ket{\Psi^{\bar
    \lambda^{c}}}
\ee
where we have used the tensor notation $\ket{\mu}=\ket{\mu^\pi}\otimes
\ket{\mu^c} $.
\end{prop}
In terms of the components, the previous Propositions tells us that 
\be
f^{\bar \lambda}_{\mu}(x_1,\dots,x_n) = t^{-\ell(\mu^{c})}  (tt_n^{-1})^{m(\mu^{c})}
f^{\bar \lambda^\pi}_{\mu^\pi}(x_1,\dots,x_n) \label{eq:f^{(s_L,s_R)}}.
\ee
\begin{proof}
The proof consist in an explicit check of eqs.(\ref{qKZC1b-s}). 
\begin{itemize}
\item Right boundary:
\begin{itemize}
\item If $\mu_n>r_L$ then $\pi\mu_n=\mu_n$ and $T_nf^{\tilde \pi
  \bar \lambda}_{\dots,\mu_n}=f^{\tilde \pi \bar
  \lambda}_{\dots,-\mu_n}$. 
\item If $r_R< \mu_n \leq r_L$ then $\pi\mu_n=0$, therefore $T_nf^{\tilde \pi
  \bar \lambda}_{\dots,0}=t_nf^{\tilde \pi \bar
  \lambda}_{\dots,0}$. On the other hand $\ell(\mu^{(c)})$ and
$m(\mu^{(c)})$ decrease by one.
\item If  $|\mu_n|\leq r_R$ then $\pi\mu_n=0$ and the equation again
follows from $T_nf^{\tilde \pi
  \bar \lambda}_{\dots,0}=t_nf^{\tilde \pi \bar
  \lambda}_{\dots,0}$.
\end{itemize}
\item Bulk
\begin{itemize}
\item If $\mu_i=\mu_{i+1}$ then $\tilde \mu_i=\tilde \mu_{i+1}$ and
  the equation follows because it is satisfied by $f^{\tilde \pi
    \bar \lambda}_{\ldots,\tilde \pi \mu_i,\tilde \pi
    \mu_{i+1},\ldots}$.
\item If $\mu_i\neq\mu_{i+1}$ and at least one of the two is in
  modulus larger $r_L$. Then the equation follows from the same
  equation satisfied by  $f^{\tilde \pi
    \bar \lambda}_{\ldots,\tilde \pi \mu_i,\tilde \pi
    \mu_{i+1},\ldots}$.
\item If  $\mu_i>\mu_{i+1}$ and $|\mu_i|,|\mu_{i+1}|\leq r_L$, then
  $\tilde \pi\mu_i= \tilde \pi\mu_{i+1}=0$ and $T_if^{\tilde \pi
    \bar \lambda}_{\ldots,\tilde \pi \mu_i=0,\tilde \pi
    \mu_{i+1}=0,\ldots}=tf^{\tilde \pi
    \bar \lambda}_{\ldots,\tilde \pi \mu_i=0,\tilde \pi
    \mu_{i+1}=0,\ldots}$. On the other hand the exchange of $\mu_i$
  and $\mu_{i+1}$ makes $\ell(\mu^{(c)})$ decrease by $1$. 
\end{itemize}
\item Left boundary
\begin{itemize}
\item If $\mu_1<r_L$ then $\pi \mu_1=m_1$ and $T_0f^{\tilde \pi
  \bar \lambda}_{\mu_1,\dots}=q^{\mu_1}t_0^{-1}f^{\tilde \pi \bar
  \lambda}_{-\mu_1,\dots}$.
\item If  $|\mu_1|\leq r_L$ then $\pi\mu_1=0$ and the equation again
follows from $T_0f^{\tilde \pi
  \bar \lambda}_{0,\dots}=t_0f^{\tilde \pi \bar
  \lambda}_{0,\dots}$.
\end{itemize}
\end{itemize}
\end{proof}

\section*{Acknowledgment}
We gratefully acknowledge support from the Australian Research Council Centre of Excellence for Mathematical and Statistical Frontiers (ACEMS). MW acknowledges support by an Australian Research Council DECRA. We warmly thank the Galileo Galilei Institute and the organisers of the research program \textit{Statistical Mechanics, Integrability and Combinatorics} for kind hospitality during part of this work. JdG would also like to thank the KITP Program \textit{New approaches to non-equilibrium and random systems: KPZ integrability, universality, applications and experiments}, supported in part by the National Science Foundation under Grant No. NSF PHY11-25915. It is a pleasure to thank Eric Rains, Ole Warnaar and Lauren Williams for discussions.
\appendix

\section{Explicit example for $n=2$}
\label{se:n=2}

To show that we can use either the dual K-matrix at the right hand side or the left hand side, we take the Yang--Baxter equation (\ref{YBeq}) and rewrite it as follows
\begin{align*}
&R_{1,2}\left(w/x\right)R_{1,3}\left(w^2\right)R_{2,3}\left(w x\right)=R_{2,3}\left(w x\right)R_{1,3}\left(w^2\right) R_{1,2}\left(w/x\right), \\
&R_{1,3}\left(w/x\right)R_{1,2}\left(w^2\right)R_{3,2}\left(w x\right)=R_{3,2}\left(w x\right)R_{1,2}\left(w^2\right) R_{1,3}\left(w/x\right), \\
&R_{1,3}\left(w/x\right)^{\tau_1}R_{3,2}\left(w x\right)(R_{1,2}\left(w^2\right)^{\tau_1})^{-1}=(R_{1,2}\left(w^2\right)^{\tau_1})^{-1}R_{3,2}\left(w x\right) R_{1,3}\left(w/x\right)^{\tau_1}.
\end{align*}
In the first line we wrote the Yang--Baxter equation, in the second we conjugated with $P_{2,3}$ both sides, then we transposed with $\tau_1$ and multiplied both sides of equation by $(R_{1,2}(w^2)^{\tau_1})^{-1}$. The final step is to transpose again with $\tau_1$ and replace the matrix $((R_{1,2}(w^2)^{\tau_1})^{-1})^{\tau_1}$ with $\widetilde{R}(w^2)$. The resulting equation can be written explicitly in a convenient way using the trace 
\begin{align}\label{YBTr}
\text{Tr}_1 (R_{1,3}\left(w/x\right)R_{3,2}\left(w x\right)\widetilde{R}_{1,2}\left(w^2\right))=\text{Tr}_1(\widetilde{R}_{1,2}\left(w^2\right) R_{3,2}\left(w x\right) R_{1,3}\left(w/x\right)).
\end{align}

Now we take (\ref{Tmat}) and write explicitly $\widetilde{K}_n(w)$ in terms of $\widetilde{R}(w^2)$ according to (\ref{Kdn1}). Using equation (\ref{YBTr}) we can push $\widetilde{R}(w^2)$ to the other boundary where it meets $K_0(w)$ giving rise to $\widetilde{K}_0(w)$ by definition (\ref{Kd01}). Here is the $n=2$ example
\begin{align*}
&\text{Tr}_{0,\bar{0}} \left(R_{0,2}\left(w x_2\right)R_{0,1}\left(w x_1\right)K_0 (w) R_{1,0}\left(w/x_1\right)R_{2,0}\left(w/x_2\right)K_n(1/w)\tilde{R}_{\bar{0},0}\left(w^2\right) P_{0,\bar{0}} \right)=
\\
&\text{Tr}_{0,\bar{0}} \left(R_{0,2}\left(w x_2\right)R_{0,1}\left(w x_1\right)K_0 (w) K_n(1/w)P_{0,\bar{0}}  R_{1,\bar{0}}\left(w/x_1\right)R_{2,\bar{0}}\left(w/x_2\right)\tilde{R}_{0,\bar{0}}\left(w^2\right) \right)=
\\
&\text{Tr}_{0,\bar{0}} \left(R_{2,\bar{0}}\left(w/x_2\right)R_{0,1}\left(w x_1\right)K_0 (w) K_n(1/w)P_{0,\bar{0}}  R_{1,\bar{0}}\left(w/x_1\right)R_{0,2}\left(w x_2\right)\tilde{R}_{0,\bar{0}}\left(w^2\right) \right)=
\\
&\text{Tr}_{0,\bar{0}} \left(R_{0,1}\left(w x_1\right)R_{2,\bar{0}}\left(w/x_2\right)K_0 (w) K_n(1/w)P_{0,\bar{0}}  R_{0,2}\left(w x_2\right)R_{1,\bar{0}}\left(w/x_1\right)\tilde{R}_{0,\bar{0}}\left(w^2\right) \right)=
\\
&\text{Tr}_{0,\bar{0}} \left(R_{1,\bar{0}}\left(w/x_1\right)R_{2,\bar{0}}\left(w/x_2\right)K_0 (w) K_n(1/w)P_{0,\bar{0}}  R_{0,2}\left(w x_2\right)R_{0,1}\left(w x_1\right)\tilde{R}_{0,\bar{0}}\left(w^2\right) \right)=
\\
&\text{Tr}_{0,\bar{0}} \left(R_{1,\bar{0}}\left(w/x_1\right)R_{2,\bar{0}}\left(w/x_2\right) K_n(1/w) R_{\bar{0},2}\left(w x_2\right)R_{\bar{0},1}\left(w x_1\right)K_0 (w) \tilde{R}_{\bar{0},0}\left(w^2\right)P_{0,\bar{0}} \right).
\end{align*}
In the first line we wrote the transfer matrix explicitly with $\tilde{K}$ as in (\ref{Kdn1}). The matrix $K_0$ acts in $V_0$ and $K_n$ in $V_{\bar{0}}$. In the second line we commuted $K_n$ and $P$ to the left. In the following three lines we used (\ref{YBTr}), rearranged R-matrices and used (\ref{YBTr}) again. This action switched the order of $M^{(1)}$ and $M^{(2)}$. In the last line we commuted $K_0$ and $P$ to the right. The last three matrices in the last line are precisely of the form (\ref{Kd01}). 

Therefore, we get
\begin{align}
&\overline{M}(w;x_1,..,x_n)=M^{(2)}(w;x_1,..,x_n) K_n(1/w) M^{(1)}(w;x_1,..,x_n), \nonumber \\
&T(w;x_1,..,x_n)=\text{Tr}_0 ( \overline{M}(w;x_1,..,x_n) \widetilde{K}_0 (w) ). 
\end{align}

\section{Commutativity of the transfer matrices}
\label{tmcomm}
Let us show that the statement 
\begin{align*}
[T(u),T(w)]=0,
\end{align*}
appearing in Theorem \ref{th:commut} holds. We need to supply the matrices entering $T(u)$ and $T(w)$, given by (\ref{Tmat}), with indices denoting the auxiliary spaces. First, write $K$ for $K_0$ and $\tilde{K}$ for $\tilde{K}_n$. Then we write the T-matrix as
\begin{align*}
T(u)=\text{Tr}_{0}\left( M_0^{(1)}(u ) K_0(u)   M_0^{(2)}\left(u\right) \tilde{K}_0(u) \right),
\end{align*}
where the subscript $0$ denotes the auxiliary space $V_0$. Take the product $T(u)T(w)$
\begin{align*}
&T(u)T(w)=\text{Tr}_{0}\left( M_0^{(1)}(u ) K_0(u)   M_0^{(2)}\left(u\right) \tilde{K}_0(u)\right) 
\text{Tr}_{\bar{0}}\left( M_{\bar{0}}^{(1)}(w) K_{\bar{0}}(w)   M_{\bar{0}}^{(2)}\left(w\right) \tilde{K}_{\bar{0}}(w)\right)= \\
&\text{Tr}_{0,\bar{0}}\left( M_0^{(1)}(u ) K_0(u)   M_0^{(2)}\left(u\right) \tilde{K}_0(u)
M_{\bar{0}}^{(1)}(w ) K_{\bar{0}}(w)   M_{\bar{0}}^{(2)}\left(w\right) \tilde{K}_{\bar{0}}(w)  \right),
\end{align*}
apply the transposition in $V_0$
\begin{align*}
\text{Tr}_{0,\bar{0}}\left(K_0(u)^{\tau _0} M^{(1)}_0(u){}^{\tau _0} \tilde{K}_0(u){}^{\tau _0}
M_0^{(2)}\left(u\right){}^{\tau _0} M^{(1)}_{\bar{0}}(w )   K_{\bar{0}}(w)   M_{\bar{0}}^{(2)}\left(w\right)  
\tilde{K}_{\bar{0}}(w) \right),
\end{align*}
the two monodromy matrices $M_0^{(2)}\left(u\right){}^{\tau _0} M^{(1)}_{\bar{0}}(w ) $ can be switched using 
\begin{align*}
R_{\bar{0},0}^{\tau_0}(w u) M_0^{(2)}\left(u\right){}^{\tau _0} M^{(1)}_{\bar{0}}(w)-
M^{(1)}_{\bar{0}}(w)  M_0^{(2)}\left(u\right){}^{\tau _0} R_{\bar{0},0}^{\tau_0}(w u).
\end{align*}
The R-matrix in this expression can be introduced using the crossing unitarity written in the form 
\begin{align*}
\tilde{R}_{\bar{0},0}^{\tau_0}(w u)R_{\bar{0},0}^{\tau_0}(w u)=\mathbb{I}. 
\end{align*}
We find
\begin{align*}
\text{Tr}_{0,\bar{0}}\left(K_0(u)^{\tau _0} M^{(1)}_0(u){}^{\tau _0} \tilde{K}_0(u){}^{\tau _0}
\tilde{R}_{\bar{0},0}^{\tau_0}(w u)
M^{(1)}_{\bar{0}}(w ) M_0^{(2)}\left(u\right){}^{\tau _0}   R_{\bar{0},0}^{\tau_0}(w u)K_{\bar{0}}(w)   M_{\bar{0}}^{(2)}\left(w\right)  
\tilde{K}_{\bar{0}}(w) \right).
\end{align*}
The next step is to transpose in $V_0$
\begin{align*}
\text{Tr}_{0,\bar{0}}\left(\tilde{R}_{\bar{0},0}(w u)  \tilde{K}_0(u)  M^{(1)}_0(u)  K_0(u)
M^{(1)}_{\bar{0}}(w ) R_{\bar{0},0}(w u) K_{\bar{0}}(w)  M_0^{(2)}\left(u\right)  M_{\bar{0}}^{(2)}\left(w\right)  
 \tilde{K}_{\bar{0}}(w) \right).
\end{align*}
The matrices $M_0^{(2)}\left(u\right)  M_{\bar{0}}^{(2)}\left(w\right)$ can be interchanged using
\begin{align*}
R_{\bar{0},0}(u/w)M_0^{(2)}\left(u\right)  M_{\bar{0}}^{(2)}\left(w\right)  -
M_{\bar{0}}^{(2)}\left(w\right)  M_0^{(2)}\left(u\right) R_{\bar{0},0}(u/w).
\end{align*}
The R-matrix in this equation can be introduced using the unitarity $\mathbb{I}=R_{0,\bar{0}}(w/u)R_{\bar{0},0}(u/w)$, hence we get
\begin{align*}
\text{Tr}_{0,\bar{0}}\bigg{(}&\tilde{R}_{\bar{0},0}(w u)  \tilde{K}_0(u)  M^{(1)}_0(u)  K_0(u)
M^{(1)}_{\bar{0}}(w ) R_{\bar{0},0}(w u) K_{\bar{0}}(w)  R_{0,\bar{0}}(w/u) \times \\
&M_{\bar{0}}^{(2)}\left(w\right)    M_0^{(2)}\left(u\right)  R_{\bar{0},0}(u/w)  \tilde{K}_{\bar{0}}(w) \bigg{)}.
\end{align*}
The cyclicity of trace allows us to take the first two matrices $\tilde{R}_{\bar{0},0}(w u)  \tilde{K}_0(u)$ to the  right side of the product
\begin{align*}
\text{Tr}_{0,\bar{0}}\bigg{(}
&M^{(1)}_0(u)  M^{(1)}_{\bar{0}}(w ) K_0(u)  R_{\bar{0},0}(w u) K_{\bar{0}}(w)  R_{0,\bar{0}}(w/u) \times\\
&M_{\bar{0}}^{(2)}\left(w\right)    M_0^{(2)}\left(u\right)   R_{\bar{0},0}(u/w)
 \tilde{K}_{\bar{0}}(w) \tilde{R}_{\bar{0},0}(w u)  \tilde{K}_0(u)\bigg{)}.
\end{align*}
The order of the K-matrices and R-matrices in the last expression can be switched by means of the reflection equations 
(\ref{ReflK0}) and (\ref{ReflDKn}) 
\begin{align*}
\text{Tr}_{0,\bar{0}}\bigg{(}
&M^{(1)}_0(u)  M^{(1)}_{\bar{0}}(w ) 
R_{\bar{0},0}(w/u)    K_{\bar{0}}(w)  R_{0,\bar{0}}(w u) K_0(u)\times\\
&M_{\bar{0}}^{(2)}\left(w\right)    M_0^{(2)}\left(u\right)   
 \tilde{K}_0(u) \tilde{R}_{0,\bar{0}}(u w)  \tilde{K}_{\bar{0}}(w) R_{0,\bar{0}}(u/w) \bigg{)}.
\end{align*}
Now we need to commute $M^{(1)}_0(u)$ to the right. In order to do that we introduce the R-matrices using the unitarity relation $\mathbb{I}=R_{\bar{0},0}(w/u)R_{0\bar{0}}(u/w)$ and use 
\begin{align*}
R_{0,\bar{0}}(u/w)M_0^{(1)}\left(u\right)  M_{\bar{0}}^{(2)}\left(w\right)  -
M_{\bar{0}}^{(2)}\left(w\right)  M_0^{(2)}\left(u\right) R_{0,\bar{0}}(u/w).
\end{align*}
The result reads 
\begin{align*}
\text{Tr}_{0,\bar{0}}\bigg{(}M^{(1)}_{\bar{0}}(w ) M^{(1)}_0(u) 
 K_{\bar{0}}(w)  R_{0,\bar{0}}(w u) K_0(u)
 M_{\bar{0}}^{(2)}\left(w\right)    M_0^{(2)}\left(u\right)   
 \tilde{K}_0(u) \tilde{R}_{0,\bar{0}}(u w)  \tilde{K}_{\bar{0}}(w)  \bigg{)}.
\end{align*}
In the last step we encountered combinations $R_{0\bar{0}}(u/w)
R_{\bar{0},0}(w/u)$ and set them to $\mathbb{I}$ by unitarity. Transposing again in $V_0$ brings $M_0^{(1)}(u)$ next to 
$M^{(2)}_{\bar{0}}(w)$
\begin{align*}
\text{Tr}_{0,\bar{0}}\bigg{(}  &\tilde{K}_{\bar{0}}(w)  M^{(1)}_{\bar{0}}(w )  K_{\bar{0}}(w)  K_0(u)^{\tau_0} \times \\
&R_{0,\bar{0}}(w u) ^{\tau_0}  M^{(1)}_0(u)^{\tau_0}  M_{\bar{0}}^{(2)}\left(w\right) \tilde{R}_{0,\bar{0}}(u w)^{\tau_0}
 \tilde{K}_0(u)^{\tau_0} M_{0}^{(2)}\left(u\right)^{\tau_0} \bigg{)}.
\end{align*}
It remains to commute $M^{(1)}$ and $M^{(2)}$ by means of $R_{0,\bar{0}}(w u) ^{\tau_0}$, use the crossing unitarity relation
\begin{align*}
\text{Tr}_{0,\bar{0}}\bigg{(} \tilde{K}_{\bar{0}}(w)  M^{(1)}_{\bar{0}}(w )  K_{\bar{0}}(w)  K_0(u)^{\tau_0}
M_{\bar{0}}^{(2)}\left(w\right)  M^{(1)}_0(u)^{\tau_0}   \tilde{K}_0(u)^{\tau_0} M_{0}^{(2)}\left(u\right)^{\tau_0} \bigg{)},
\end{align*}
and transpose in $V_0$
\begin{align*}
\text{Tr}_{0,\bar{0}}\bigg{(} \tilde{K}_{\bar{0}}(w)  M^{(1)}_{\bar{0}}(w )  K_{\bar{0}}(w) M_{\bar{0}}^{(2)}\left(w\right)  
\tilde{K}_0(u) M^{(1)}_0(u) K_0(u)  M_{0}^{(2)}\left(u\right)  \bigg{)}.
\end{align*}
The order of the T-matrices is now switched, so we arrive at $T(w)T(u)$.

\end{document}